\documentclass[12pt]{article}

\usepackage{mathtools,amssymb,amsthm}
\usepackage{bm}
\usepackage{natbib}
\usepackage[dvipdf]{graphicx}
\usepackage{epsfig}
\usepackage{color}
\usepackage{latexsym}
\usepackage{setspace}
\usepackage{hyperref}
\usepackage{booktabs}

\usepackage{graphicx}
\graphicspath{{figures/}}
\usepackage{caption}
\usepackage{subcaption}
\usepackage{authblk}
\usepackage{algorithm}
\usepackage{algpseudocode}

\newtheorem{theorem}{Theorem}
\newtheorem{lemma}{Lemma}
\newtheorem{corollary}{Corollary}
\newtheorem{proposition}{Proposition}

\newcommand{\bA}{\mathbf{A}}

\newcommand{\bB}{\mathbf{B}}
\newcommand{\bb}{\mathbf{b}}

\newcommand{\be}{\mathbf{e}}

\newcommand{\bI}{\mathbf{I}}

\newcommand{\bv}{\mathbf{v}}

\newcommand{\bX}{\mathbf{X}}
\newcommand{\bx}{\mathbf{x}}

\newcommand{\by}{\mathbf{y}}

\newcommand{\cE}{\mathcal{E}}

\newcommand{\cN}{\mathcal{N}}

\newcommand{\cU}{\mathcal{U}}

\newcommand{\sT}{\mathsf{T}}

\newcommand{\balpha}{\bm{\alpha}}
\newcommand{\bbeta}{\bm{\beta}}
\newcommand{\bdelta}{\bm{\delta}}

\newcommand{\bgamma}{\bm{\gamma}}

\newcommand{\bmu}{\bm{\mu}}

\newcommand{\bSigma}{\bm{\Sigma}}

\newcommand{\bpsi}{\mbox{\boldmath $\psi$}}

\newcommand{\bzero}{\boldsymbol{0}}

\newcommand{\rP}{\mathrm{P}}

\newcommand{\diag}{\mathrm{diag}}

\newcommand{\R}{\mathbb{R}}

\newcommand{\argmax}{\operatornamewithlimits{argmax}}
\newcommand{\argmin}{\operatornamewithlimits{argmin}}

\DeclarePairedDelimiter\parens{\lparen}{\rparen}

\DeclarePairedDelimiter\braces{\lbrace}{\rbrace}
\numberwithin{equation}{section}
\renewcommand{\baselinestretch}{1.12}

\title{Quadratic Discriminant Analysis by Projection}
\date{}
\author{Ruiyang Wu and Ning Hao \hspace{.2cm}\\
    Department of Mathematics, University of Arizona}

\begin{document}
\maketitle

\begin{abstract}
  Discriminant analysis, including linear discriminant analysis (LDA)
  and quadratic discriminant analysis (QDA), is a popular approach to
  classification problems. It is well known that LDA is suboptimal to
  analyze heteroscedastic data, for which QDA would be an ideal
  tool. However, QDA is less helpful when the number of features in a
  data set is moderate or high, and LDA and its variants often perform
  better due to their robustness against dimensionality. In this work,
  we introduce a new dimension reduction and classification method
  based on QDA\@. In particular, we define and estimate the optimal
  one-dimensional (1D) subspace for QDA, which is a novel hybrid
  approach to discriminant analysis. The new method can handle data
  heteroscedasticity with number of parameters equal to that of
  LDA\@. Therefore, it is more stable than the standard QDA and works
  well for data in moderate dimensions.  We show an estimation
  consistency property of our method, and compare it with LDA, QDA,
  regularized discriminant analysis (RDA) and a few other competitors
  by simulated and real data examples.
\end{abstract}
\noindent%

\section{Introduction}\label{sec:introduction}

Discriminant analysis is a standard tool for classification. For
example, LDA and QDA aim to find hyperplanes and quadratic
hypersurfaces, respectively, to separate the data points. LDA is one
of the most popular techniques for classification because of its
simplicity and robustness against growing
dimensionality. Nevertheless, the performance of LDA relies on the
equal covariance assumption. In contrast, QDA allows data
heteroscedasticity. The cost of the flexibility is to estimate more
parameters of the QDA model, which requires a large sample size. To
make the QDA approach more robust, \cite{friedman1989regularized}
proposed regularized discriminant analysis (RDA), which shrinks the
separate covariances of different classes toward a common pooled
covariance that can be further shrunken to a diagonal matrix when
necessary. The level of shrinkage is controlled by tuning parameters,
which are often tuned by cross-validation. As a compromise between LDA
and QDA, RDA is a successful classification tool which has been
further developed in~\cite{guo2007regularized}.

Based on Fisher's original idea \citep{fisher1936use}, LDA aims to
find a 1D projection which best separates the data. Fisher suggested
the direction that maximizes the ratio of between-class variance to
within-class variance. Under the Gaussian and equal covariance
assumption, the population version of LDA rule, or PoLDA for short, is
the optimal classification rule. This implies two facts. First, there
is no information loss to project the data onto the PoLDA
direction. Second, PoLDA minimizes classification error. These
properties of PoLDA do not hold under data heteroscedasticity. In
general, it is impossible to project the data to a 1D subspace without
loss of information. Even if a good projection exists, QDA might be a
better choice than LDA to separate the projected data. This motivates
us to study the optimal 1D projection for heteroscedastic data. To
elaborate, we will define an optimal direction in which the projected
data are separated by QDA with least classification error. We propose
an algorithm to approximate this optimal direction and show its
consistency. With strong heteroscedasticity, our method can outplay
LDA methods. Because the number of parameters in our algorithm is
similar to that of LDA, our method needs a smaller sample size than
QDA does.

In this work, we are mainly concerned with data sets with $p<n<p^2$
where $n$ is the sample size and $p$ is the number of features. In
this situation, LDA may be seriously biased because of data
heteroscedasticity, and QDA is not stable numerically due to
dimensionality. Our method offers an alternative classification tool
for practitioners.  We have to point out that, in the last 20 years,
there are many works on discriminant analysis for high and ultra-high
dimensional data, \cite{li2015sparse}, \cite{jiang2018direct},
\cite{wu2019quadratic}, \cite{gaynanova2019sparse}, just to name a few
closely related to QDA\@. We refer readers to two review papers
\citep{mai2013review,qin2018review} for more comprehensive summaries
of recent developments. A majority of these works are based on various
sparsity assumptions. In contrast, we do not impose sparsity
assumptions and our method shares invariance property with the
classical LDA and QDA methods. Therefore, we do not suggest to apply
our method to high dimensional data directly.

\section{Classification by 1D Projection}
\label{sec:class-1d-proj}
\subsection{Optimal 1D Projection for Heteroscedastic Gaussian data}
\label{sec:opt-1d-proj}

Let \(\bX\) be a \(p\)-dimensional random vector, and \(Y\in\{0,1\}\)
be its class label with \(\bX|\{Y=k\} \sim \cN(\bmu_k,\bSigma_k)\),
and \(\rP(Y=k) = \pi_k\), \(k = 0,1\), where \(\bmu_k\)'s are
\(p\)-dimensional vectors and \(\bSigma_k\)'s are \(p\) by \(p\)
symmetric positive definite matrices. Define
$\bSigma=\pi_0\bSigma_0+\pi_1\bSigma_1$, which is the weighted average
of within-class covariances. Note that $\bSigma$ is the common
within-class covariance for homoscedastic data, and when
$\bSigma_0 \neq \bSigma_1$, it is the expectation of estimated
within-class covariance under a misspecified homoscedastic model. We
assume $\pi_0=\frac12$ in this paper for easy presentation.

For a heteroscedastic Gaussian model with known parameters, the QDA
rule is optimal in a sense that it minimizes the classification error
for any $\bX\in\mathbb{R}^p$. It labels an observation $\bX=\bx^*$ by
class 1 when
\begin{multline}\label{qdarule}
  {\bx^*}^{\mathsf{T}}\parens*{\bSigma_0^{-1} - \bSigma_1^{-1}} \bx^*
  -2{\bx^*}^{\mathsf{T}} \parens*{\bSigma_0^{-1}\bmu_0 - \bSigma_1^{-1}\bmu_1}\\
  +\bmu_0^{\mathsf{T}}\bSigma_0^{-1}\bmu_0 -
  \bmu_1^{\mathsf{T}}\bSigma_1^{-1}\bmu_1+\log\parens*{|\bSigma_0|/|\bSigma_1|}
  >0.
\end{multline}

The discriminant boundary of the QDA rule is a quadratic hypersurface,
determined by $p(p+3)/2$ parameters. In practice, when $p$ is moderate
or large, it is difficult to estimate the quadratic boundary
accurately due to the large parameter space. While dimension reduction
is a plausible approach to reduce the number of parameters, it is
impossible to reduce the sample space without loss of classification
power for general covariances $\bSigma_0$ and $\bSigma_1$. In
contrast, LDA assumes equal covariance
\(\bSigma=\bSigma_0 = \bSigma_1\), under which the QDA
rule~\eqref{qdarule} reduces to the LDA rule
\begin{equation}\label{ldarule}
  -{\bx^*}^{\mathsf{T}} \bSigma^{-1}(\bmu_0 - \bmu_1) + \frac12 (\bmu_0 + \bmu_1)^{\mathsf{T}}\bSigma^{-1}(\bmu_0 - \bmu_1) > 0.
\end{equation} 
In this special case, the optimal discriminant boundary is a
hyperplane with the normal vector
\begin{equation}
  \label{eq:1}
  \bbeta=\bSigma^{-1}(\bmu_0-\bmu_1).
\end{equation}
The classification error of the optimal rule is 
\begin{equation}
  \Phi\parens*{-\sqrt{(\bmu_0-\bmu_1)^{\sT}\bSigma^{-1}(\bmu_0-\bmu_1)}/2}=\Phi\parens*{-\sqrt{\bbeta^{\sT}\bSigma\bbeta}/2},    
\end{equation}
where $\Phi(\cdot)$ is the cumulative distribution function (CDF) of a
standard normal random variable. Note that for any nonzero vector
$\balpha$, the distribution of \(\balpha^{\sT}\bX|\{Y=k\}\) is
$\cN(\balpha^{\sT}\bmu_k,\balpha^{\sT}\bSigma\balpha)$. It is
straightforward to derive that the LDA rule in the direction $\balpha$
can achieve classification error
$\Phi(-|\balpha^{\mathsf{T}}(\bmu_0-\bmu_1)|/(2\sqrt{\balpha^{\mathsf{T}}\bSigma\balpha}))$,
with a minimal value \(\Phi(-\sqrt{\bbeta^{\sT}\bSigma\bbeta}/2)\)
when $\balpha=c\bbeta$ for any $c\ne0$. In summary, the direction of
$\bbeta$ is the optimal direction to project the data to achieve the
best classification accuracy. More importantly, we won't lose any
classification power after dimension reduction to this 1D
subspace. This is one of the reasons that the LDA-based approach is
more popular than QDA in data analysis. For the downside, LDA is
suboptimal when the data is heteroscedastic. First of all, the LDA
direction, calculated by the same formula
$\bbeta=\bSigma^{-1}(\bmu_0-\bmu_1)$, with
$\bSigma=\pi_0\bSigma_0+\pi_1\bSigma_1$, is not the best direction to
project the data onto. In a special case when \(\bmu_0 = \bmu_1\) and
\(\bSigma_0 = \bI\) and \(\bSigma_1 = \diag\{2,1,\dots,1\}\), the
direction \(\be_1=(1,0,\dots,0)^{\sT}\) is the best, but
$\bbeta=\bzero$. An estimator to $\bbeta$ would give a random and
uninformative direction. Second, even if the best direction is known,
the 1D LDA is outperformed by the 1D QDA after the projection, when
the marginal variances are not equal. While the second issue is minor
and easy to fix, in this paper, we aim to define and estimate the
optimal 1D projection for heteroscedastic Gaussian data.

For a classification rule $\Psi:\mathbb{R}^p\to\{0,1\}$, the
classification error is defined by $\rP(\Psi(\bX)\ne Y)$. Let $E_0$ be
the classification error of the QDA rule defined in~\eqref{qdarule},
and $E_{LDA}$ be the classification error of LDA~\eqref{ldarule} with
$\bSigma=\pi_0\bSigma_0+\pi_1\bSigma_1$ under
heteroscedasticity. Under a projection of $\bX$ to a 1D subspace
spanned by $\balpha$, define $E(\balpha)$ and $E_{LDA}(\balpha)$ by
the classification errors of the QDA and LDA rules for the projected
data. Then we have
\begin{equation}
  \label{eq:2}
  E_0\leq \min_{\balpha\ne\bzero}E(\balpha)\leq \min_{\balpha\ne\bzero}E_{LDA}(\balpha)\leq E_{LDA}.
\end{equation}
The equal signs in (\ref{eq:2}) hold in the special case when
$\bSigma_0=\bSigma_1$. In general cases, it is impossible to approach
$E_0$ empirically if $p^2>n$. Nevertheless, it is easier to estimate
the direction that minimizes $E(\balpha)$. We show an explicit formula
for $E(\balpha)$ in the following theorem.

\begin{theorem}
  \label{thr:1}
  Assume \(\bX|\{Y=k\} \sim \cN(\bmu_k,\bSigma_k)\),
  $\pi_k=P(Y=k)=\frac12$, \(k = 0,1\). Let
  \(m_k = \balpha^{\sT} \bmu_k\),
  \(\sigma_k^2 = \balpha^{\sT} \bSigma_k \balpha\), where \(k = 0,1\),
  \(\balpha \in\R^p \setminus \{\bzero\}\). Then the classification
  error function for 1D QDA \(E \colon \R^p \backslash \{0\} \to \R\)
  in (\ref{eq:2}) satisfies
  \begin{equation}
    \label{eq:3}
    E(\balpha)=
    \begin{dcases}
      \Phi \parens*{-\frac{|m_0 - m_1|}{2\sigma}} & \sigma_0 =
      \sigma_1 := \sigma,\\
      \begin{multlined}[0.6\linewidth]
        \frac12 + \frac{1}{2} \Phi \parens*{{\frac{\sigma_1(m_1 - m_0) -
              \sigma_0 \sqrt{\Delta}}{\sigma_0^2 - \sigma_1^2}}}\\
        -\frac{1}{2} \Phi \parens*{{\frac{\sigma_1(m_1 - m_0) +
              \sigma_0 \sqrt{\Delta}}{\sigma_0^2 - \sigma_1^2}}}\\
        +\frac{1}{2} \Phi \parens*{{\frac{\sigma_0(m_1 - m_0) +
              \sigma_1 \sqrt{\Delta}}{\sigma_0^2 - \sigma_1^2}}}\\
        -\frac{1}{2} \Phi \parens*{{\frac{\sigma_0(m_1 - m_0) -
              \sigma_1 \sqrt{\Delta}}{\sigma_0^2 - \sigma_1^2}}}
      \end{multlined}
      & \sigma_0 \neq \sigma_1,
    \end{dcases}
  \end{equation}
  where
  \(\Delta = (m_0 - m_1)^2 + (\sigma_0^2 - \sigma_1^2) \log(\sigma_0^2
  / \sigma_1^2)\).
\end{theorem}

We define a direction
$\balpha_0\in \argmin_{\balpha\neq\bzero} E(\balpha)$ an optimal
direction for 1D QDA\@. The following proposition summarizes two
well-known special cases when close-form solutions can be derived to
minimize $E(\balpha)$.
\begin{proposition}
  \label{prop:1}
  Under the assumptions in Theorem~\ref{thr:1}, the following results
  hold.
  \begin{enumerate}
  \item If \(\bSigma_0 = \bSigma_1=\bSigma\),
    \begin{align}
      \label{eq:4}
      \balpha_0 = \argmin_{\balpha\neq \bzero} E(\balpha) =
      \argmax_{\balpha\ne\bzero} \parens*{\frac{\balpha^{\sT}
      (\bmu_0 - \bmu_1)(\bmu_0 - \bmu_1)^{\sT}
      \balpha}{\balpha^{\sT} \bSigma \balpha}}=\bSigma^{-1}(\bmu_1-\bmu_0).
    \end{align}
  \item If \(\bmu_0 = \bmu_1\),
    \begin{align}
      \label{eq:5}
      \balpha_0 = \argmin_{\balpha\neq\bzero} E(\balpha) =
      \argmax_{\balpha\ne\bzero} \left(\max
      \braces*{\frac{\balpha^{\sT} \bSigma_1
      \balpha}{\balpha^{\sT} \bSigma_0
      \balpha},\frac{\balpha^{\sT} \bSigma_0
      \balpha}{\balpha^{\sT} \bSigma_1 \balpha}}\right).
    \end{align}
    If there is a unique maximum among all the eigenvalues of
    \(\bSigma_0^{-1}\bSigma_1\) and \(\bSigma_1^{-1}\bSigma_0\), then
    \(\balpha_0\) is the eigenvector corresponding to the greatest
    eigenvalue.
  \end{enumerate}
\end{proposition}

When the number of features is moderate, e.g., $p<n<p^2$, the standard
QDA is not stable empirically. As an alternative approach, we attempt
to estimate the best 1D subspace for dimension reduction before
conducting QDA\@. Intuitively, this approach is more robust than the
standard QDA because much fewer parameters are needed. In particular,
it requires $p-1$ parameters for direction estimation and two more
parameters for the quadratic boundary after projection. Thus the total
number of parameters is similar to that of LDA\@. As a result, our
method performs similarly to LDA for homoscedastic data, and it is
more sensitive to data heteroscedasticity than the LDA approach.

The LDA direction in (\ref{eq:4}) is well-defined and unique up to a
scalar when $\bmu_0\ne\bmu_1$. However, the optimal direction to 1D
QDA might not be unique, especially when some symmetric structure
occurs in the model. For example, in (\ref{eq:5}), if
$\bSigma_0=c_0\bI$ and $\bSigma_1=c_1\bI$ with $c_0\ne c_1$, every
direction is an optimal direction because of symmetry. In general, the
optimal direction would be unique up to a scalar although it is
difficult to specify the exact conditions on uniqueness.

\subsection{Method and computation}

Let \(\{\bx^i_k : 1 \leq i \leq n_k\}\) be i.i.d.\ observations from
\(\bX|\{Y=k\}\), \(k=0,1\). With Theorem~\ref{thr:1}, we can
approximate the classification error \(E(\balpha)\) with
\(\hat{E}(\balpha)\), which is equation~\eqref{eq:3} plugged in by
sample means \(\hat{\bmu}_0\) \(\hat{\bmu}_1\) and sample covariance
matrices \(\hat{\bSigma}_0\) \(\hat{\bSigma}_1\). We then find the
minimizing direction \(\hat{\balpha}_0\) of
\(\hat{E}(\balpha)\). After projecting all the training data and test
data to \(\hat{\balpha}_0\), predictions are made based on the 1D QDA
rule of the projected data. We call this procedure QDA by projection
(QDAP), which is summarized in Algorithm~\ref{alg:1}.

\begin{algorithm}
  \caption{QDA by Projection (QDAP)}
  \label{alg:1}
  \begin{algorithmic}[1]
    \State
    \(\hat{\balpha}_0 \gets \argmin_{\balpha\neq\bzero}
    \hat{E}(\balpha)\)
    \State
    \(x^i_k \gets \hat{\balpha} {}_0^{\sT} \bx^i_k\),
    \(k = 0,1\); \(i = 1,\dots,n_k\)
    \State
    \(\hat{\phi}(x) \gets \mbox{1D QDA rule derived with projected
      data } \{x^i_k\}_{k=0,1}^{i=1,\dots,n_k}\)
    \State
    \Return
    \(\hat{y} \gets \hat{\phi}(x) \mbox{ for any new observation }
    \bx\), where \(x = \hat{\balpha} {}_0^{\sT} \bx\)
  \end{algorithmic}
\end{algorithm}

By Corollary~\ref{cor:2} and~\ref{cor:3} in the appendix,
$\hat{E}(\balpha)$ is smooth almost everywhere, and it is uniformly
continuous when viewed as a function defined on the unit sphere. Thus
the existence of minimizer is guaranteed by the compactness of the
unit sphere. We implemented a coordinate descent algorithm to conduct
the optimization. Proposition~\ref{prop:1} provides two good initial
directions, i.e.,~\eqref{eq:4} and \eqref{eq:5} to warm start the
coordinate descent algorithm. The implementation details are
illustrated in Appendix~\ref{sec:coord-desc-algor}.

\section{Theoretical Properties}
\label{sec:theor-prop}

Proposition~\ref{prop:1} shows that LDA is a special case of our
method in the population level. Thanks to the explicit formula, it is
straightforward to see that the LDA direction can be consistently
estimated. The following theorem shows a counterpart result for the 1D
QDA\@. As a by-product, it implies our method performs similar to LDA
under the equal covariance assumption.
\begin{theorem}\label{thr:2}
  Assume that \(\bX|\{Y=k\} \sim \cN(\bmu_k,\bSigma_k)\), \(k =
  0,1\). Let \(\{\bx^i_k : i \geq 1\}\) be a sequence of i.i.d.\
  observations from \(\bX|\{Y=k\}\),
  \(\smash{\hat{\bmu} {}_0^n,\hat{\bmu} {}_1^n,\hat{\bSigma}
    {}_0^n,\hat{\bSigma} {}_1^n}\) be sample means and sample
  covariance matrices calculated with first \(n\) observations in each
  class, and \(\hat{E}^n(\balpha)\) be the empirical classification
  error, i.e.~\eqref{eq:3} with the previous estimates plugged in. If
  \(E(\balpha)\) has a unique minimizer
  \[\balpha_0 = \argmin\limits_{\balpha \in \mathbb{P}^{p-1}} E(\balpha)\]
  and assume that
  \[\hat{\balpha}_0^n \in \argmin\limits_{\balpha \in \mathbb{P}^{p-1}} \hat{E}^n(\balpha)\]
  then
  \[\hat{\balpha} {}_0^n \xrightarrow{a.s.} \balpha_0 \; \text{ as } n \to \infty.\]
\end{theorem}

Since the classification error function \(E\) depends only on the
direction of vectors in \(\R^p\backslash\{0\}\), it is essentially a
function defined on the \(p - 1\) dimensional real projective space
\(\mathbb{P}^{p - 1}\), which consists of all one dimensional
subspaces of \(\R^p\). (see Corollary~\ref{cor:3} and \ref{cor:4} in
appendix for details). Practically, we may simply view $\balpha$ as a
unit vector up to a sign. To make the theorem mathematically rigorous,
we use \(\mathbb{P}^{p - 1}\) as the domain of $\balpha$. It is
standard in mathematics to denote the one dimensional subspace spanned
by a vector \(\balpha\) by equivalent class \([\balpha]\). But we will
omit the brackets for easy presentation whenever there is no
ambiguity.

LDA and QDA share an invariance property, which ensures that the
classification result is unaffected by any invertible affine
transformation of the data. To elaborate, if we apply the same
nonsingular linear transformation to the training data and future test
data, the prediction results of LDA and QDA will not change. The
following proposition indicates that the invariance property also
holds for our method.
\begin{proposition}
  \label{prop:2}
  For \(k = 0, 1\), let \(\{\bx^i_k : 1 \leq i \leq n_k\}\) be i.i.d
  observations from \(\bX|\{Y=k\}\), and
  \(\tilde\bx^i_k = \bb + \bA\bx^i_k\), where \(\bb \in \R^p\),
  \(\bA\) is a \(p\) by \(p\) full rank matrix. Let
  \(\hat{\balpha}_0\) (\(\hat{\tilde\balpha}_0\)) be the unique (up to
  a scalar) minimizer in step 1 of Algorithm~\ref{alg:1}, with
  \(\hat{E}\) (\(\hat{\tilde E}\)) derived from training data
  \(\{\bx^i_k\}\) (\(\{\tilde\bx^i_k\}\)). Then the following equation
  holds:
  \[\hat{\tilde\balpha}_0=c\parens*{\bA^{\sT}}^{-1}\hat{\balpha}_0,\]
  where \(c\) is a nonzero constant.
\end{proposition}

This implies
\(\tilde x^i_k=\hat{\tilde\balpha}^{\sT}_0 \tilde\bx^i_k =
c\hat{\balpha} {}_0^{\sT} \bA^{-1} \bb + c\hat{\balpha} {}_0^{\sT}
\bx^i_k = c\hat{\balpha}{}_0^{\sT} \bA^{-1} \bb + cx^i_k\), where
$x^i_k$ is the projected data defined in Algorithm~\ref{alg:1}, step
2. That is, the projected data before and after transformation,
$x^i_k$ and $\tilde x^i_k$ are up to an affine transformation. It
implies

\begin{corollary}
  \label{cor:1}
  Algorithm~\ref{alg:1} is invariant under invertible affine
  transformations.
\end{corollary}

Here is a remark on the Gaussian assumption before we move on to the
numerical studies. The formulation~\eqref{eq:3} of the classification
error of QDA with respect to direction $\balpha$ relies on the
Gaussian distribution. As a consequence, the definition of the optimal
projection, $\balpha_0$, depends on the Gaussian assumption. Without
the Gaussian assumption, the direction $\balpha_0$ is still defined as
the minimizer of~\eqref{eq:3}, although it might not the be the
optimal projection in the sense of minimizing expected classification
error. This is analogous to the story for LDA\@. Without the Gaussian
assumption, LDA still works and is consistent to its population
version, although the population version of LDA is not the Bayesian or
optimal rule any more. In our case, the main theoretical results,
i.e., consistency (Theorem~\ref{thr:2}) and invariance
(Proposition~\ref{prop:2}) still hold without the Gaussian assumption.

\section{Numerical Studies}
\label{sec:numerical-studies}
\subsection{Method for Comparison}
In this section, we compare our method, Algorithm~\ref{alg:1} (QDAP),
with LDA, DSDA~\citep{mai2012direct}, QDA,
DAP~\citep{gaynanova2019sparse}, and RDA~\citep{guo2007regularized} by
both simulated and real data examples. Besides the classical methods
LDA and QDA, RDA is a well known regularization approach which works
well for moderate and high dimensional data. DSDA and DAP are two
representatives of modern high dimensional classification tools. For
DSDA, DAP and RDA, we used the R packages provided by the authors with
default settings. For LDA and QDA, we used functions from R
recommended package `MASS'. In simulated data examples, the oracle
method that employs the true model for prediction is included for
comparison as a benchmark.

\subsection{Simulated data}
\label{sec:simulated-data}

We illustrate seven data generation settings as follows. In the first
five models, the data are generated from Gaussian distributions with
parameters specified below.

\begin{itemize}
\item Model 1: \(\bSigma_0 = \bSigma_1 = \bI_{p}\).
  \(\bmu_0 = \bm0_{p}\), \(\bmu_1 = \frac13\bm1_{p}\).
\item Model 2:
  \(\bSigma_0 = \bSigma_1 = \bB^{\sT} \bB + {\mathrm{diag}}(\bv)\),
  where \(\bB\) is a \(p \times p\) matrix with IID entries from
  \(\cN (0,1)\) distribution, and \(\bv\) is a \(p \times 1\) vector
  with IID entries from \(\cU(0,1)\) distribution.
  \(\bmu_0 = \bm0_{p}\), \(\bmu_1 = \bm1_{p}\).
\item Model 3: \(\bSigma_0 = \bI_{p}\), \(\bSigma_1 = (\sigma_{ij})\),
  where \(\sigma_{ii} = 3\) and \(\sigma_{ij} = 2\) for \(i \neq
  j\). \(\bmu_0 = \bm0_{p}\), \(\bmu_1 = \bm1_{p}\).
\item Model 4: Same settings as Model 3 except that
  \(\bmu_1 = \bm0_{p}\).
\item Model 5: Same settings as Model 3 except that
  \(\bSigma_0 = \mathrm{diag}(10, \bm1_{p - 1})\), and \(\bmu_1\) has
  IID entries from \(\cN (0,1/p)\) distribution.
\end{itemize}
In the next two models, the data are from multivariate
$t$-distributions with 3 degrees of freedom \(t_3(\bmu_k, \bSigma_k)\)
\citep{anderson2003introduction}.
\begin{itemize}
\item Model 6: Same \(\bSigma_k\)'s and \(\bmu_k\)'s as Model 2.
\item Model 7: Same \(\bSigma_k\)'s and \(\bmu_k\)'s as Model 5.
\end{itemize}

The number of features is set to \(p = 50\). In each model, sample
sizes are set to \(n = 200,300,400,500\) and 600 for training, with
\(n/2\) samples in each class. A test set with 500 observations in
each class is used for calculating classification errors. In
Tables~\ref{tab:simulation-1}-\ref{tab:simulation-7}, we report the
average classification errors (in percentage) with standard errors,
based on 100 replicates for each scenario. In models 2, 5, 6 and 7,
the model parameters are generated once, and all replicates are
independently generated from the same model.
\begin{table}
  \footnotesize \centering
  \begin{tabular}{rlllllll}
    \toprule
    $n$ & LDA & QDA & RDA & DSDA & DAP & QDAP & Oracle\\
    \midrule
    200 & 17.41 (0.18) & 35.79 (0.25) & 14.22 (0.16) & 17.59 (0.17) & 19.04 (0.17) & 17.46 (0.18) & 11.89 (0.10)\\
    300 & 15.37 (0.14) & 31.37 (0.20) & 13.41 (0.13) & 15.68 (0.14) & 16.74 (0.14) & 15.42 (0.14) & 11.93 (0.10)\\
    400 & 14.63 (0.13) & 28.79 (0.19) & 13.08 (0.10) & 15.00 (0.13) & 15.95 (0.15) & 14.65 (0.13) & 11.86 (0.11)\\
    500 & 14.07 (0.11) & 26.39 (0.18) & 12.84 (0.09) & 14.29 (0.12) & 15.04 (0.12) & 14.06 (0.11) & 11.72 (0.10)\\
    600 & 13.64 (0.12) & 24.63 (0.17) & 12.74 (0.11) & 13.92 (0.12) & 14.52 (0.13) & 13.67 (0.12) & 11.90 (0.11)\\
    \bottomrule
  \end{tabular}
  \caption{Average classification errors in percentage (with standard errors in parenthesis) for model 1.}
  \label{tab:simulation-1}
\end{table}

\begin{table}
  \footnotesize \centering
  \begin{tabular}{rlllllll}
    \toprule
    $n$ & LDA & QDA & RDA & DSDA & DAP & QDAP & Oracle\\
    \midrule
    200 & 9.11 (0.14) & 26.61 (0.27) & 9.49 (0.16) & 9.66 (0.17) & 26.52 (0.44) & 9.24 (0.14) & 5.31 (0.07)\\
    300 & 7.67 (0.10) & 20.22 (0.22) & 7.85 (0.10) & 8.04 (0.11) & 20.06 (0.35) & 7.69 (0.10) & 5.27 (0.06)\\
    400 & 6.98 (0.09) & 16.85 (0.17) & 7.18 (0.11) & 7.28 (0.10) & 17.30 (0.32) & 6.99 (0.09) & 5.30 (0.07)\\
    500 & 6.59 (0.08) & 14.91 (0.15) & 6.71 (0.07) & 6.80 (0.08) & 15.29 (0.26) & 6.57 (0.08) & 5.24 (0.06)\\
    600 & 6.30 (0.08) & 13.47 (0.14) & 6.41 (0.08) & 6.53 (0.08) & 13.72 (0.22) & 6.32 (0.08) & 5.32 (0.06)\\
    \bottomrule
  \end{tabular}
  \caption{Average classification errors in percentage (with standard errors in parenthesis) for model 2.}
  \label{tab:simulation-2}
\end{table}

\begin{table}
  \footnotesize \centering
  \begin{tabular}{rlllllll}
    \toprule
    $n$ & LDA & QDA & RDA & DSDA & DAP & QDAP & Oracle\\
    \midrule
    200 & 36.92 (0.26) & 28.38 (0.23) & 18.12 (0.15) & 24.66 (0.27) & 15.20 (0.18) & 17.16 (0.26) & 7.94 (0.08)\\
    300 & 35.02 (0.25) & 25.29 (0.20) & 18.23 (0.13) & 23.29 (0.24) & 12.34 (0.16) & 11.76 (0.16) & 8.10 (0.09)\\
    400 & 33.14 (0.28) & 23.62 (0.19) & 18.27 (0.13) & 22.70 (0.21) & 11.31 (0.15) & 10.41 (0.13) & 8.20 (0.10)\\
    500 & 31.44 (0.24) & 21.80 (0.14) & 18.12 (0.13) & 21.92 (0.15) & 10.65 (0.12) & 9.63 (0.09) & 8.21 (0.09)\\
    600 & 30.60 (0.24) & 20.43 (0.14) & 18.16 (0.13) & 21.41 (0.16) & 10.02 (0.10) & 9.18 (0.08) & 8.06 (0.07)\\
    \bottomrule
  \end{tabular}
  \caption{Average classification errors in percentage (with standard errors in parenthesis) for model 3.}
  \label{tab:simulation-3}
\end{table}

\begin{table}
  \footnotesize \centering
  \begin{tabular}{rlllllll}
    \toprule
    $n$ & LDA & QDA & RDA & DSDA & DAP & QDAP & Oracle\\
    \midrule
    200 & 49.88 (0.16) & 30.59 (0.22) & 46.44 (0.48) & 49.61 (0.17) & 25.02 (0.77) & 19.53 (0.26) & 10.10 (0.08)\\
    300 & 50.42 (0.15) & 27.56 (0.20) & 46.33 (0.44) & 49.80 (0.17) & 20.39 (0.80) & 13.93 (0.16) & 9.91 (0.08)\\
    400 & 50.17 (0.18) & 25.82 (0.16) & 46.23 (0.43) & 49.88 (0.18) & 18.45 (0.72) & 12.41 (0.12) & 9.93 (0.09)\\
    500 & 49.95 (0.17) & 24.23 (0.17) & 47.41 (0.36) & 49.67 (0.15) & 18.02 (0.84) & 11.71 (0.11) & 10.16 (0.10)\\
    600 & 50.09 (0.16) & 23.03 (0.13) & 47.73 (0.32) & 50.04 (0.15) & 19.90 (1.15) & 11.18 (0.10) & 9.96 (0.09)\\
    \bottomrule
  \end{tabular}
  \caption{Average classification errors in percentage (with standard errors in parenthesis) for model 4.}
  \label{tab:simulation-4}
\end{table}

\begin{table}
  \footnotesize \centering
  \begin{tabular}{rlllllll}
    \toprule
    $n$ & LDA & QDA & RDA & DSDA & DAP & QDAP & Oracle\\
    \midrule
    200 & 35.82 (0.22) & 22.36 (0.19) & 35.61 (0.27) & 35.98 (0.26) & 23.07 (0.62) & 19.72 (0.25) & 7.31 (0.08)\\
    300 & 34.69 (0.22) & 18.69 (0.15) & 34.44 (0.24) & 34.83 (0.25) & 20.33 (0.66) & 14.03 (0.17) & 7.29 (0.08)\\
    400 & 32.98 (0.17) & 16.82 (0.13) & 32.81 (0.19) & 33.21 (0.19) & 18.63 (0.65) & 12.53 (0.12) & 7.14 (0.08)\\
    500 & 32.65 (0.16) & 15.58 (0.13) & 32.39 (0.17) & 32.39 (0.16) & 19.06 (0.62) & 11.89 (0.11) & 7.42 (0.08)\\
    600 & 32.00 (0.13) & 14.69 (0.11) & 31.79 (0.15) & 32.07 (0.16) & 19.57 (0.70) & 11.37 (0.10) & 7.25 (0.07)\\
    \bottomrule
  \end{tabular}
  \caption{Average classification errors in percentage (with standard errors in parenthesis) for model 5.}
  \label{tab:simulation-5}
\end{table}

\begin{table}
  \footnotesize \centering
  \begin{tabular}{rlllllll}
    \toprule
    $n$ & LDA & QDA & RDA & DSDA & DAP & QDAP & Oracle\\
    \midrule
    200 & 8.69 (0.12) & 23.54 (0.26) & 9.07 (0.15) & 9.09 (0.15) & 24.06 (0.37) & 8.78 (0.12) & 5.22 (0.07)\\
    300 & 7.40 (0.10) & 19.23 (0.19) & 7.67 (0.11) & 7.72 (0.11) & 18.36 (0.29) & 7.45 (0.10) & 5.09 (0.06)\\
    400 & 6.86 (0.09) & 16.44 (0.14) & 7.03 (0.10) & 7.16 (0.11) & 14.81 (0.26) & 6.90 (0.09) & 5.06 (0.07)\\
    500 & 6.25 (0.08) & 14.78 (0.16) & 6.49 (0.08) & 6.46 (0.09) & 13.43 (0.24) & 6.28 (0.08) & 4.96 (0.07)\\
    600 & 6.23 (0.08) & 13.57 (0.14) & 6.37 (0.09) & 6.37 (0.09) & 11.68 (0.19) & 6.23 (0.08) & 5.13 (0.06)\\
    \bottomrule
  \end{tabular}
  \caption{Average classification errors in percentage (with standard errors in parenthesis) for model 6.}
  \label{tab:simulation-6}
\end{table}

\begin{table}
  \footnotesize \centering
  \begin{tabular}{rlllllll}
    \toprule
    $n$ & LDA & QDA & RDA & DSDA & DAP & QDAP & Oracle\\
    \midrule
    200 & 32.36 (0.20) & 22.57 (0.24) & 31.60 (0.25) & 32.08 (0.24) & 23.72 (0.55) & 22.54 (0.26) & 6.27 (0.09)\\
    300 & 30.48 (0.18) & 19.62 (0.20) & 29.79 (0.18) & 30.19 (0.19) & 22.05 (0.48) & 17.72 (0.16) & 6.28 (0.08)\\
    400 & 29.57 (0.17) & 17.74 (0.20) & 28.97 (0.19) & 29.50 (0.18) & 21.30 (0.45) & 16.08 (0.14) & 6.36 (0.07)\\
    500 & 28.52 (0.16) & 16.48 (0.18) & 27.97 (0.15) & 28.19 (0.16) & 21.10 (0.44) & 15.43 (0.13) & 6.22 (0.08)\\
    600 & 28.08 (0.14) & 16.03 (0.21) & 27.52 (0.14) & 27.99 (0.14) & 20.87 (0.45) & 14.77 (0.11) & 6.18 (0.07)\\
    \bottomrule
  \end{tabular}
  \caption{Average classification errors in percentage (with standard errors in parenthesis) for model 7.}
  \label{tab:simulation-7}
\end{table}

For models 1 and 2, the LDA assumption of equal covariance matrices is
satisfied. LDA performs well, and our method performs similarly to
LDA\@. RDA performs better than LDA for model 1, due to the diagonal
covariance structure. For models 3 and 4, the data are
heteroscedastic, and there is only one useful direction for
classification. As a result, our method (QDAP) performs the best. The
LDA-based methods performs much worse due to the unequal covariance
structure. The standard QDA suffers from small sample sizes. DAP
method performs reasonably well and ranks in the second place. Model 5
represents a more general heteroscedastic setting. In this case, our
method is suboptimal to QDA if the sample size is big enough. However,
our method could outperform QDA when the sample size is moderate, due
to the bias-variance trade-off. As a result, our method performs best
in Table~\ref{tab:simulation-5} for all sample sizes in the given
range. To demonstrate the robustness of our algorithm for non-Gaussian
data, we consider models 6 and 7, which are similar to models 2 and 5
except that multivariate $t$ distributions with 3 degrees of freedom
are employed. In model 6, LDA performs the best, while our method
performs similarly to LDA in terms of both classification error and
its standard error. For model 7, three QDA-based methods are better
than LDA-based methods and our method achieves the best accuracy. We
conclude from these two examples that our method is similar to LDA and
other methods in terms of robustness to heavy-tailed data.

\subsection{Real Data}
\label{sec:real-data}

In this subsection, five real data sets are used to compare these
classification methods. In each real data experiment, we randomly
assigned 60\% of the observations into the training set and the rest
into the test set. We randomly split each real data set 300 times, and
calculated average classification error along with its standard
error.

\subsubsection{Breast Cancer Wisconsin Data Set}
\label{sec:breast-canc-wisc}

The breast cancer data set, created by Dr.\ WIlliam
H. Wolberg~\citep{wolberg1990multisurface}, is available on the UCI
Machine Learning Repository~\citep{Dua:2019}. There are $n=699$
instances of patients from Dr.\ Wolberg's clinical cases. 10 features
are recorded for each patient, $p=9$ of which are the explanatory
variables. The 10th feature assigns the patients into two classes ---
``benign'' and ``malignant''.

\subsubsection{Ultrasonic Flowmeter Diagnostics Data Set}
\label{sec:ultr-flowm-diagn}

This data set, provided by~\citet{gyamfi2018linear}, is available on
the UCI Machine Learning Repository~\citep{Dua:2019}. The goal of this
data set is to predict the health status of some flowmeters installed
at UK using diagnostic data. There are $n=87$ instances of diagnosed
flowmeters and the diagnostic data comes in $p=36$ dimensions. Two
classes are either ``Healthy'' or ``Installation effects''.

\subsubsection{Heart Disease Data Set}
\label{sec:heart-disease-data}

This data set, provided by Andras Janosi, William Steinbrunn, Matthias
Pfisterer and Robert Detrano, is available on the UCI Machine Learning
Repository~\citep{Dua:2019}. There are \(n=303\) patients in
total. $p=13$ different attributes are used to predict the patients'
angiographic disease status, which could be either 0 (\(<\) 50\%
diameter narrowing) or 1 (\(>\) 50\% diameter narrowing).

\subsubsection{Image Segmentation Data}
\label{sec:image-segm-data}

This data set, created by Vision Group, University of Massachusetts,
is available on the UCI Machine Learning
Repository~\citep{Dua:2019}. There are 2310 total images in 7
different classes, with \(330\) images each. To make this a binary
classification problem, we only include class 1 (brickface) and 4
(cement) for analysis. There are 19 features in total. Features 1, 3,
4, 5 are almost constants within the chosen classes, so they were
removed from the data, leaving $p=15$ features for classification.

\subsubsection{Satellite Data Set}
\label{sec:satellite-data-set}
This data set, provided by Ashwin Srinivasan, is available on the UCI
Machine Learning Repository~\citep{Dua:2019}. Satellite images are
labeled into 9 classes. Only class 1 (red soil) and class 3 (grey
soil) are considered for our analysis, where there are 1072 images in
class 1 and 961 images in class 3. $p=36$ attributes (9 pixels times 4
spectral bands) are used for classification.

\subsubsection{Results}
\label{sec:real-data-results}
Average classification errors (in percentage) for these experiments
are summarized in Table~\ref{tab:real-data-analysis}. LDA performs
reasonably well for all data sets, but our method outplays LDA with a
margin, especially in the first two data sets. To better understand
the result, we performed classical Box's M test \citep{box1949general}
and a modern high dimensional two-sample covariance test proposed by
\citet{cai2013two}. All the \(p\) values for the 5 data sets are below
\(2.68\times 10^{-8}\), indicating strong evidence of
heteroscedasticity. Nevertheless, the original QDA suffers from low
sample sizes, and in particular, fails to work in data sets 2 and
4. As a QDA based method, our method is more versatile and gives
better classification results. It outperforms both LDA and QDA\@. RDA
performs well except in data set 2. DSDA and DAP, as representatives
of sparse methods for high dimensional data, produce slightly worse
results than LDA and our method. Overall, our method performs the best
among the algorithms in comparison.

\begin{table}
  \footnotesize \centering
  \begin{tabular}{rllllll}
    \toprule
    & LDA & QDA & RDA & DSDA & DAP & QDAP\\
    \midrule
    Data set 1 & 4.62 (0.06) & 5.02 (0.07) & 4.23 (0.06) & 4.87 (0.06) & 4.24 (0.06) & 3.30 (0.04)\\
    Data set 2 & 1.58 (0.11) & NA & 34.05 (0.38) & 2.94 (0.26) & 15.52 (0.41) & 0.89 (0.08)\\
    Data set 3 & 17.81 (0.17) & 20.86 (0.18) & 17.56 (0.18) & 18.00 (0.17) & 18.43 (0.19) & 17.48 (0.17)\\
    Data set 4 & 0.72 (0.02) & NA & 0.78 (0.03) & 0.84 (0.03) & 1.64 (0.04) & 0.69 (0.02)\\
    Data set 5 & 1.37 (0.02) & 1.79 (0.03) & 1.38 (0.02) & 1.39 (0.02) & 1.54 (0.02) & 1.32 (0.02)\\
    \bottomrule
  \end{tabular}
  \caption[Real data analysis summary]{Average classification errors in percentage (with standard errors in parenthesis) for different classification methods. Data set 1: Breast cancer Wisconsin data set. Data set 2: Ultrasonic flowmeter diagnostics data set. Data set 3: Heart disease data set. Data set 4: Image segmentation data set. Data set 5: Satellite data set.}
  \label{tab:real-data-analysis}
\end{table}

\section{Discussion}\label{sec:discussion}
In this work, we propose a new dimension reduction and classification
method based on QDA\@. The empirical studies show that our algorithm
performs well for data sets with moderate dimensions and unequal
covariance structures. An R package \texttt{QDAP} implementing our
algorithm is available on
\url{https://github.com/ywwry66/QDA-by-Projection-R-Package}. Note
that we assume equal prior probability in this paper for easy
presentation, without which all theoretical results still hold with
minor modifications. Moreover, the implementation in our R package
does not rely on this assumption.

We discuss here briefly a few related works in the literature. In
particular, \cite{gaynanova2019sparse} proposes a quadratic
classification rule via linear dimension reduction called DAP, which
works for high dimensional classification with unequal
covariances. Roughly speaking, DAP estimates simultaneously two
directions \(\bpsi_0 = \bSigma_0^{-1}\bdelta\) and
\(\bpsi_1 = \bSigma_1^{-1}\bdelta\) where
\(\bdelta = \bmu_0 - \bmu_1\), and then employs QDA for classification
after projecting the data to these two directions. Empirically, a
sparse method is used for estimating \(\bpsi_0\) and \(\bpsi_1\). In
the population level, the space spanned by \(\bpsi_0\) and \(\bpsi_1\)
can be very different from or even orthogonal to our 1D optimal
subspace spanned by \(\balpha_0 = \arg \min_{\balpha} E(\balpha)\). In
short, DAP does not aim to find such an optimal projection. An
advantage of DAP is that it conducts variable selection and works for
high dimensional data. It is an interesting research direction to
extend our method in a sparse high dimensional setting. Some recent
works~\citep{cannings2017random,tian2021rase} propose to ensemble
classifiers on random subspaces. Instead of searching for an optimal
projection, these works employ and combine a collection of classifiers
on subspaces, which may perform better when a single optimal
projection does not exist. In practice, an asymptotic expansion of the
classification error would be helpful to decide sample sizes for
training~\citep{kharin2013robustness}. It is an interesting research
direction to study such an expansion for our method. Last but not
least, it is momentous to study classification with dependent
observations, for example, time series
data~\citep{krafty2016discriminant}, spatially correlated
data~\citep{li2020high}, and clipping of random
field~\citep{de2000bayesian}.

\section*{Acknowledgement} The authors are grateful to the Associate
Editor and two referees for helpful comments. This work was supported
by the National Science Foundation Grant DMS-1722691 and CCF-1740858;
and Simons Foundation Grant 524432.

\renewcommand{\baselinestretch}{1} 
\appendix
\section{Appendix}
\label{sec:appendix}

\subsection{Proof of Theorem~\ref{thr:1}}
\label{sec:proof-theorem-1}
We prove Theorem~\ref{thr:1} in this appendix. Let \(\psi_{\balpha}\)
be the 1D Bayesian rule for \((\balpha^{\mathsf{T}} \bX,
Y)\). Clearly,
\(\balpha^{\mathsf{T}} \bX|\{Y=k\} \sim
\mathcal{N}(m_k,\sigma_k^2)\). We prove by 2 cases:
\begin{enumerate}
\item \(\sigma_0 \neq \sigma_1\). Without loss of generality, we may
  assume \(\sigma_0 > \sigma_1\). In this case,
  \[
    \psi_{\balpha}(x) = 1_{\{x\colon q(x) > 0\}}(x) = 1_{(r_1, r_2)}(x), 
  \]
  where
  \[
    q(x) = \parens*{\frac1{\sigma_0^2}-\frac1{\sigma_1^2}}x^2 -
    2\parens*{\frac{m_0}{\sigma_0^2}-\frac{m_1}{\sigma_1^2}}x +
    \parens*{\frac{m_0^2}{\sigma_0^2}-\frac{m_1^2}{\sigma_1^2}} +
    \log\parens*{\frac{\sigma_0^2}{\sigma_1^2}} > 0
  \] is the 1D version of QDA rule~\eqref{qdarule}, and
  \(r_1,r_2=(\parens*{m_1\sigma_0^2-m_0\sigma_1^2}\pm
  \sigma_0\sigma_1\sqrt{\Delta})/(\sigma_0^2-\sigma_1^2)\) with
  \(\Delta=(m_0-m_1)^2+(\sigma_0^2-\sigma_1^2)\log(\sigma_0^2/\sigma_1^2)\)
  are the roots of \(q(x)\).

  The classification error is calculated as follows.
  \begin{align*}
    E(\balpha)=&\frac12\rP(\psi_{\balpha}(\balpha^{\mathsf T}\bX)=1|Y=0)+\frac12\rP(\psi_{\balpha}(\balpha^{\mathsf T}\bX)=0|Y=1)\\
    =&\frac12\rP(r_1<\balpha^{\mathsf T}\bX<r_2|Y=0)+\frac12\rP(\balpha^{\mathsf T}\bX<r_1\textrm{ or }\balpha^{\mathsf T}\bX>r_2|Y=1)\\
    =&\frac12\rP\parens*{\frac{r_1-m_0}{\sigma_0}<\frac{\balpha^{\mathsf T}\bX-m_0}{\sigma_0}<\frac{r_2-m_0}{\sigma_0}\bigg|Y=0}\\
               &+\frac12\rP\parens*{\frac{\balpha^{\mathsf T}\bX-m_1}{\sigma_1}<\frac{r_1-m_1}{\sigma_1}\textrm{ or }\frac{\balpha^{\mathsf T}\bX-m_1}{\sigma_1}>\frac{r_2-m_1}{\sigma_1}\bigg|Y=1}\\
    =&\frac12\Phi\parens*{\frac{r_2-m_0}{\sigma_0}}-\frac12\Phi\parens*{\frac{r_1-m_0}{\sigma_0}}+\frac12\Phi\parens*{\frac{r_1-m_1}{\sigma_1}}+\frac12\parens*{1-\Phi\parens*{\frac{r_2-m_1}{\sigma_1}}}.
  \end{align*}
  This is exactly the expression of \(E(\balpha)\) in
  Theorem~\ref{thr:1} when \(\sigma_0\neq\sigma_1\).
\item \(\sigma_0=\sigma_1=\sigma\). In this case \(\psi_{\balpha}\)
  reduces to the 1D LDA rule. Assuming \(m_0>m_1\),
  \(\psi_{\balpha}(x) = 1_{(-\infty, (m_0+m_1)/2)}(x)\). So
  \begin{align*}
    E(\balpha)=&\frac12\rP(\psi_{\balpha}(\balpha^{\mathsf T}\bX)=1|Y=0)+\frac12\rP(\psi_{\balpha}(\balpha^{\mathsf T}\bX)=0|Y=1)\\
    =&\frac12\rP\parens*{\balpha^{\mathsf T}\bX<\frac{m_0+m_1}2\bigg|Y=0}+\frac12\rP\parens*{\balpha^{\mathsf T}\bX>\frac{m_0+m_1}2\bigg|Y=1}\\
    =&\frac12\rP\parens*{\frac{\balpha^{\mathsf T}\bX-m_0}{\sigma}<\frac{m_1-m_0}{2\sigma}\bigg|Y=0}+\frac12\rP\parens*{\frac{\balpha^{\mathsf T}\bX-m_1}{\sigma}>\frac{m_0-m_1}{2\sigma}\bigg|Y=1}\\
    =&\frac12\Phi\parens*{\frac{m_1-m_0}{2\sigma}}+\frac12\parens*{1-\Phi\parens*{\frac{m_0-m_1}{2\sigma}}}\\
    =&\Phi\parens*{-\frac{|m_1-m_0|}{2\sigma}}.
  \end{align*}
  Similarly, we can show the same formula for \(m_0<m_1\). When
  \(m_0=m_1\), LDA becomes random guess so \(E(\balpha)=1/2\), which
  is again the same as function value \(\Phi(0)\).
\end{enumerate}

\subsection{Continuity and Analyticity of Classification Error Function}
\label{con-dif-err}
We present a few properties of the classification error function
$E(\balpha)$ which are helpful in the proof of Theorem~\ref{thr:2}.

Assuming \(r(\balpha) = (m_0 - m_1)/\sigma_1\) and
\(g(\balpha) = \sigma_0/\sigma_1\), we can rewrite the classification
error \(E(\balpha)\) as the composition of
\(\cE \colon \R \times \R_{>0} \to \R\) and
\((r(\balpha), g(\balpha))\), where
\begin{equation}
  \label{eq:7}
  \cE(r, g) =
  \begin{dcases}
    \Phi \parens*{-\frac{|r|}{2}} & g = 1\\
    \begin{multlined}[0.6\linewidth]
      \frac12 + \frac{1}{2} \Phi \parens*{{\frac{r - g
            \sqrt{\Delta}}{g^2 - 1}}} -\frac{1}{2} \Phi
      \parens*{{\frac{r +
            g \sqrt{\Delta}}{g^2 - 1}}}\\
      +\frac{1}{2} \Phi \parens*{{\frac{rg + \sqrt{\Delta}}{g^2 - 1}}}
      -\frac{1}{2} \Phi \parens*{{\frac{rg - \sqrt{\Delta}}{g^2 - 1}}}
    \end{multlined}
    & g \neq 1, g > 0
  \end{dcases}
\end{equation}
\(\Delta = r^2 + (g^2 - 1)\log (g^2)\).
\begin{proposition}
  \label{prop:3}
  The following properties hold for \(\cE\):
  \begin{enumerate}
  \item \(\forall (r, g)\in \R \times \R_{>0}\),
    \(\cE(r, g)\in (0, 1/2]\), 
  \item \(\cE\) is continuous,
  \item \(\cE\) is analytic on
    \(\R \times (\R_{ > 0} \backslash \{1\})\).
  \end{enumerate}
\end{proposition}

\begin{proof}
  \begin{enumerate}
  \item We prove this by two cases:
    \begin{enumerate}
    \item If \(g = 1\), since \(0 < \Phi(-|r|/2) \leq \Phi(0) = 1/2\),
      \(\cE(r, g) = \Phi(-|r|/2) \in (0, 1/2]\). 
    \item If \(g\neq 1\), we can rewrite \(\cE\) as
      \begin{align*}
        \cE(r, g) = &\frac12 + \frac12(\Phi(c_1)-\Phi(c_2))+\frac12(\Phi(d_1)-\Phi(d_2))\\
                    & \frac12 + \frac12(\Phi(c_1)-\Phi(d_2))+\frac12(\Phi(d_1)-\Phi(c_2)), 
      \end{align*}
      where \(c_1=(r-g\sqrt{\Delta})/(g^2 - 1)\),
      \(c_2=(rg-\sqrt{\Delta})/(g^2 - 1)\),
      \(d_1=(rg+\sqrt{\Delta})/(g^2 - 1)\),
      \(d_2=(r+g\sqrt{\Delta})/(g^2 - 1)\).

      Since
      \(\sqrt{\Delta} = \sqrt{r^2 + (g^2 - 1)\log (g^2)} > \sqrt{r^2}
      = |r|\), we have
      \(c_1 - c_2 = -(\sqrt{\Delta} + r)/(g + 1) < 0\),
      \(d_1 - d_2 = -(\sqrt{\Delta} - r)/(g + 1) < 0\), which implies
      \(\Phi(c_1)-\Phi(c_2) < 0\) and \(\Phi(d_1)-\Phi(d_2) <
      0\). Thus, \(\cE(r, g) < 1/2\).

      To prove \(\cE(r, g) > 0\), we investigate separately for
      \(0 < g < 1\) and \(g > 1\). When \(0 < g < 1\),
      \(c_1 - d_2 = -2g\sqrt{\Delta}/(g^2 - 1) > 0\), so
      \(\Phi(c_1)-\Phi(d_2) > 0\), and
      \(\cE(r, g) > 1/2 + (1/2)0 + (1/2)(0 - 1) = 0\). When \(g > 1\),
      we can prove \(\Phi(d_1)-\Phi(c_2) > 0\) and get
      \(\cE(r, g) > 0\) as well.

      Combining these two inequalities, we have
      \(\cE(r, g)\in(0, 1/2)\).
    \end{enumerate}
  \item Let \(U=\R \times (\R_{ > 0} \backslash \{1\})\), then
    \(U^c=\mathbb{R} \times \{1\}\). \(\cE\) restricted on \(U\) is
    continuous because it is a composition of continuous
    functions. Similarly, \(\cE\) restricted on \(U^c\) is also
    continuous. Since \(U\) is an open subset of
    \(\R \times \R_{>0}\), \(\cE\) is continuous at every point of
    \(U\). Thus, we only need to prove \(\cE\) is continuous at every
    point of \(U^c\).

    For any \((\rho, 1)\in U^c\), it suffices to show
    \(\lim_{U \ni (r, g) \to (\rho, 1)}\cE(r, g)=\cE(\rho, 1)\). There
    are three cases:
    \begin{enumerate}
    \item If \(\rho = 0\), then for any
      \((r, g)\in U\)
      \begin{align*}
        \left|\cE(r, g)-\frac12\right|=&\left|\frac12(\Phi(c_1)-\Phi(c_2))+\frac12(\Phi(d_1)-\Phi(d_2))\right|\\
        \leq &\left|\frac12(\Phi(c_1)-\Phi(c_2))\right|+\left|\frac12(\Phi(d_1)-\Phi(d_2))\right|\\
        \leq & \frac{L}2|c_1-c_2|+\frac{L}2|d_1-d_2|\\
      \end{align*}
      The last inequality holds because \(\Phi\) is Lipschitz
      continuous. Since \(|c_1-c_2|=|\sqrt{\Delta} + r|/(g + 1)\to 0\)
      and \(|d_1-d_2|=|\sqrt{\Delta} - r|/(g + 1)\to 0\) as
      \((r, g) \to (0, 1)\) in \(U\), we have
      \(\lim_{U \ni (r, g) \to (0, 1)}\cE(r, g)=1/2 = \cE(0, 1)\).
    \item If \(\rho > 0\), as \((r, g) \to (\rho, 1)\) in \(U\),
      \[
        |d_1 - d_2| = \frac{\left|\sqrt{\Delta} - r\right|}{g + 1} =
        \frac{(g^2 - 1) \log(g^2)}{(g + 1)\left|\sqrt{\Delta} +
            r\right|} \to 0, 
      \]
      so
      \(\lim_{U \ni (r, g) \to (\rho, 1)}|\Phi(d_1) - \Phi(d_2)| \to
      0\) by Lipschitz continuity of \(\Phi\).

      For any \((r, g) \in U\),
      \begin{align*}
        \Phi(c_1) = &\Phi\parens*{\frac{r - g\sqrt{\Delta}}{g^2 - 1}}\\
        = &\Phi\parens*{\frac{r - rg}{g^2 - 1} + \frac{rg - g\sqrt{\Delta}}{g^2 - 1}}\\
        = &\Phi\parens*{-\frac{r}{g + 1} - g \frac{\Delta - r^2}{(g^2 - 1)\parens*{\sqrt{\Delta} + r}}}\\
        = &\Phi\parens*{-\frac{r}{g + 1} - g \frac{\log(g^2)}{\sqrt{\Delta} + r}} \to \Phi\parens*{-\frac{|\rho|}{2}}
      \end{align*}
      when \((r, g) \to (\rho, 1)\). Similar arguments yield
      \(\lim_{U \ni (r, g) \to (\rho, 1)} \Phi(c_2) =
      \Phi(|\rho|/2)\). As a result,
      \begin{align*}
        \lim_{U \ni (r, g) \to (\rho, 1)}\cE(r, g) = &\lim_{U \ni (r, g) \to (\rho, 1)} \frac12(\Phi(c_1)-\Phi(c_2))+\frac12(\Phi(d_1)-\Phi(d_2))+\frac12\\
        = &\frac{1}{2}\parens*{\Phi\parens*{-\frac{|\rho|}{2}} - \Phi\parens*{\frac{|\rho|}{2}}} + \frac{1}{2}\\
        = &\Phi\parens*{-\frac{|\rho|}{2}}\\
        = &\cE(\rho, 1)
      \end{align*}
    \item For \(\rho < 0\), by a similar argument to the last case, we
      have
      \(\lim_{U \ni (r, g) \to (\rho, 1)}\cE(r, g) = \cE(\rho, 1)\).
    \end{enumerate}
  \item Clearly, \(\R \times (\R_{ > 0} \backslash \{1\})\) is an open
    subset of \(\R \times \R_{ > 0}\). \(\cE\) is analytic on
    \(\R \times (\R_{ > 0} \backslash \{1\})\) because it is a
    composition of analytic functions.
  \end{enumerate}
\end{proof}

The properties of \(\cE\) have direct implications on the properties
of \(E\). The next corollary presents a few of them.
\begin{corollary}
  \label{cor:2}
  The following results hold for
  \(E \colon \R^p \backslash \{0\} \to \R\):
  \begin{enumerate}
  \item \(\forall \balpha \neq 0, E(\balpha) \in (0,1/2]\), 
  \item E is continuous,
  \item \(E\) is analytic Lebesgue a.e.
  \end{enumerate}
\end{corollary}

One important property of \(E\) is homogeneity of degree 0, i.e.\
\(E(c\balpha) = E(\balpha)\) for any \(c \neq 0\), which is easy to
see by definition~\eqref{eq:3}. This allows us to characterize \(E\)
with function \(E'\colon \mathbb{P}^{p - 1}\to\R\) through the
factorization \(E = E' \circ Q\), where \(E'\) is defined as
\[
  E'([\balpha]) = E(\balpha), 
\]
and \(Q\colon\R^p\backslash \{0\}\to \mathbb{P}^{p - 1}\) is the
canonical projection:
\[
  Q(\balpha) = [\balpha].
\]
\begin{corollary}
  \label{cor:3}
  \(E'\) is a well-defined uniformly continuous function.
\end{corollary}
\begin{proof}
  If \([\balpha] = [\bbeta]\), then \(\balpha = c\bbeta\) for some
  \(c\neq 0\). Thus,
  \(E'([\balpha]) = E(\balpha) = E(c\bbeta) = E(\bbeta) =
  E'([\bbeta])\). This proves \(E'\) is well-defined.

  \(\mathbb{P}^{p - 1}\) is endowed with the quotient topology induced
  by \(Q\), that is, \(U\subseteq \mathbb{P}^{p - 1}\) is open iff
  \(Q^{-1}(U)\subseteq \mathbb{R}^p\backslash\{0\}\) is open. For any
  \(V\subseteq\mathbb{R}\), \(E^{-1}(V) = Q^{-1}(E'^{-1}(V))\) is open
  since \(E\) is continuous. As a result, \(E'^{-1}(V)\) must be open
  as well. This proves \(E'\) is continuous.

  Since \(\mathbb{P}^{p - 1}\) is compact, we conclude \(E'\) is
  uniformly continuous by Heine–Cantor theorem.
\end{proof}

With the help of \(E'\) we can prove the following property of \(E\):
\begin{corollary}
  \label{cor:4}
  \(\argmin_{\balpha} E(\balpha)\) is non-empty.
\end{corollary}
\begin{proof}
  Since \(E'\colon \mathbb{P}^{p - 1}\to\R\) is continuous and its
  domain is compact, \(\argmin E'\) is non-empty by
  Extreme Value Theorem.

  Assume \([\bbeta]\in \argmin E'\) and
  \(\bgamma\) is arbitrary element of \(\R^{p}\backslash\{0\}\), then
  \(E(\bbeta) = E'([\bbeta])\leq E'([\bgamma]) = E(\bgamma)\). So
  \(\bbeta \in \argmin E\).
\end{proof}

The proof of Corollary~\ref{cor:4} shows how we can translate a
property of \(E'\) directly to a property of \(E\). In practice, this
is often possible. With some abuse of notation, it is beneficiary to
identify \(E\) with \(E'\), and write \([\balpha]\) just as
\(\balpha\). With this in mind, we can think of \(E\) as a uniformly
continuous function defined on projective space \(\mathbb{P}^{p-1}\).

\subsection{Proof of Theorem~\ref{thr:2}}
\label{sec:proof-theorem-2}
We show consistency of our algorithm in this section. We start with a
few lemmas. Denote by \(f_n\rightrightarrows f\) if $f_n$ is uniformly
convergent to $f$.
\begin{lemma}
  \label{lem:1}
  Let \(S\) be a set, \(X\), \(Y\) be metric spaces. Assume
  \(f,f_n\colon S\to X\), \(g\colon X\to Y\),
  \(f_n\rightrightarrows f\). If \(g\) is uniformly continuous, then
  \(g\circ f_n\rightrightarrows g\circ f\).
\end{lemma}
\begin{proof}
  Let \(d_X\) and \(d_Y\) be the metrics on \(X\) and \(Y\)
  respectively. For any \(\epsilon>0\), there exists a \(\delta>0\),
  such that whenever \(d_X(x_1,x_2)\leq\delta\),
  \(d_Y(g(x_1),g(x_2))\leq\epsilon\). For this \(\delta\), there
  exists an \(N>0\), such that whenever \(n\geq N\),
  \(d_X(f(s),f_n(s))\leq \delta\) for all \(s\in S\), thus
  \(d_Y(g\circ f(s),g\circ f_n(s))\leq\epsilon\) for all \(s\in S\).
\end{proof}
\begin{lemma}
  \label{lem:2}
  Let \(X\), \(Y\), \(Z\) be metric spaces. Assume
  \(f,f_n\colon X\to Y\), \(g\colon Y\to Z\),
  \(f_n\rightrightarrows f\). If \(X\) is compact, \(Y\) is complete,
  \(f\), \(f_n\) and \(g\) are all continuous, then
  \(g\circ f_n\rightrightarrows g\circ f\).
\end{lemma}
\begin{proof}
  Let \(I = f(X)\cup(\bigcup_{n = 1}^{\infty}f_n(X))\), we first show
  \(I\) is totally bounded.

  Since \(X\) is compact, \(f\) is continuous, it must also be
  uniformly continuous. For any \(\epsilon > 0\), there exists
  \(\delta > 0\), such that whenever \(d_X(x, x') < \delta\),
  \(d_Y(f(x), f(x')) < \epsilon/2\). Let \(B_x(\delta)\) be open balls
  centered at \(x\) with radius \(\delta\), then
  \(X\subset \bigcup_{x\in X} B_x(\delta)\). By compactness of \(X\),
  $X$ is covered by finite number of those balls, say,
  \(X\subset \bigcup_{i = 1}^{N_1} B_{x_i}(\delta)\). Because
  \(f_n\rightrightarrows f\), there exists \(N_2 > 0\), such that when
  \(n > N_2\), \(d_Y(f(x), f_n(x)) < \epsilon/2\) for any \(x\in X\).

  We now claim that
  \(f(X)\cup(\bigcup_{n = N_2 + 1}^{\infty}f_n(X))\subset \bigcup_{i =
    1}^{N_1} B_{f(x_i)}(\epsilon)\). To see this, for any \(x\in X\),
  there is an \(i_0\in\{1, \dots, N_1\}\) such that
  \(d_X(x, x_{i_0}) < \delta\), thus
  \(d_Y(f(x), f(x_{i_0})) < \epsilon/2\). Moreover, if \(n > N_2\),
  \(d_Y(f(x), f_n(x)) < \epsilon/2\), so
  \(d_Y(f_n(x), f(x_{i_0}))\leq d_Y(f(x), f(x_{i_0})) + d_Y(f(x),
  f_n(x)) < \epsilon\). This proves that
  \(f(X)\cup(\bigcup_{n = N_2 + 1}^{\infty}f_n(X))\) is covered by
  finite \(\epsilon\)-balls.

  \(\bigcup_{n = 1}^{N_2}f_n(X)\) is compact and totally bounded
  because it is finite union of compact sets. As a result, it can also
  be covered by finite \(\epsilon\)-balls. Combining these two
  collections of \(\epsilon\)-balls, we have found a finite cover of
  \(I\). Thus \(I\) is totally bounded.

  Since \(Y\) is complete, \(\bar I\), the closure of \(I\), must be
  complete and totally bounded, and thus compact. We can restrict
  \(g\) to \(\bar I\) such that \(g|_{\bar I}\) becomes uniformly
  continuous. Obviously, \(g\circ f_n = g|_{\bar I}\circ f_n\),
  \(g\circ f = g|_{\bar I}\circ f\). By Lemma~\ref{lem:1}, we have
  \(g\circ f_n\rightrightarrows g\circ f\).

\end{proof}
\begin{lemma}
  \label{lem:3}
  $f$ and $\{f_n\}_{n=1}^{\infty}$ are functions on a compact metric
  space $X$. Assume that $f$ is continuous, and has a unique minimizer
  \(x_*=\argmin_X f\). If \(f_n\rightrightarrows f\), then
  \(x_n\to x_*\), where \(x_n\in\argmin_X f_n\).
\end{lemma}
\begin{proof}
  Suppose \(x_n\not\to x_*\), then there exists an open ball \(B\)
  centered at \(x_*\), and a subsequence \(x_{n(m)}\subset
  B^c\). Since \(X\) is compact, we can further find a subsequence
  \(x_{n(m(l))}\) and \(t\in X\) such that \(x_{n(m(l))}\to
  t\). \(B^c\) is closed, thus \(t\in B^c\) and \(t\neq x_*\). For any
  \(\epsilon>0\), there is \(l_0>0\), such that whenever
  \(l\geq l_0\), \(|f(x)-f_{n(m(l))}(x)|\leq \epsilon/2\) for all
  \(x\in X\). So
  \(f(x_{n(m(l))})\leq f_{n(m(l))}(x_{n(m(l))})+\epsilon/2\leq
  f_{n(m(l))}(x_*)+\epsilon/2\leq f(x_*) +\epsilon\). This yields
  \(f(t)=f(\lim_l x_{n(m(l))})=\lim_l f(x_{n(m(l))})\leq f(x_*)\),
  which contradicts with the uniqueness of global minimizer of
  \(f\). Thus, we can conclude \(x_n\to x_*\).
\end{proof}

\noindent
\textbf{Proof of Theorem~\ref{thr:2}:}

By strong Law of Large Numbers, we have
\(\hat\bmu^n_k\xrightarrow{a.s.}\bmu_k\) and
\(\hat\bSigma^n_k\xrightarrow{a.s.}\bSigma_k\). By Egorov's theorem,
for any \(i\in \mathbb{N}\), there exists an event \(\Omega_i\) such
that \(P(\Omega_i^c) < 1/i\), and
\(\hat\bmu^n_k(\omega)\rightrightarrows\bmu_k\) and
\(\hat\bSigma^n_k(\omega)\rightrightarrows\bSigma_k\) for \(k=0,1\) on
\(\Omega_i\), where \(\R^p\) is equipped with Euclidean norm
\(\|\cdot\|_2\) and \(\R^{p\times p}\) is equipped with Frobenius norm
\(\|\cdot\|_F\). Let \(\lambda^{\mathrm{min}}_k\) be the smallest
eigenvalue of \(\bSigma_k\), and
\(\lambda = \min\{\lambda^{\mathrm{min}}_0, \lambda^{\mathrm{min}}_1\}
> 0\). There exists an integer \(N > 0\) such that whenever
\(n \geq N\),
\(\|\bSigma_k - \hat\bSigma^n_k(\omega)\|_F \leq \lambda/2\) for any
\(\omega\in\Omega_i\). From now on, we shall fix an
\(\omega\in\Omega_i\), and omit ``\(\omega\)'' for easy presentation.

Consider the following subsequences
\(\hat\bmu^{n(l)}_k = \hat\bmu^{N + l}_k\) and
\(\hat\bSigma^{n(l)}_k = \hat\bSigma^{N + l}_k\), \(k = 0, 1\). We
want to show that
\(\hat E^{n(l)}(\balpha)=\mathcal{E}((\hat m_0^{n(l)} - \hat
m_1^{n(l)})/\hat\sigma_1^{n(l)},
\hat\sigma_0^{n(l)}/\hat\sigma_1^{n(l)})\) converges uniformly to
\(E(\balpha)=\mathcal{E}((m_0 - m_1)/\sigma_1, \sigma_0/\sigma_1)\) on
\(\balpha\in \mathbb{S}^{p-1}\), where
\(m_k(\balpha)=\balpha^{\mathsf T}\bmu_k\),
\(\sigma_k(\balpha)=\sqrt{\balpha^{\mathsf T}\bSigma_k\balpha}\),
\(\hat m_k^n(\balpha)=\balpha^{\mathsf T}\hat\bmu{}_k^n\),
\(\hat\sigma^n_k(\balpha)=\sqrt{\balpha^{\mathsf
    T}\hat\bSigma{}_k^n\balpha}\), and \(\mathcal{E}\) is defined as
in Appendix~\ref{con-dif-err}. We also use \(\|\cdot\|_2\) to denote
the matrix operator norm induced by Euclidean norm. For any
\(\balpha \in \mathbb{S}^{p - 1}\),
\begin{align*}
  \left|\balpha^{\mathsf T}\bSigma_k\balpha-\balpha^{\mathsf T}\hat\bSigma_k^{n(l)}\balpha\right| & \leq \|\balpha\|_2\left\|\parens*{\bSigma_k-\hat\bSigma_k^{n(l)}}\balpha\right\|_2\\
                                                                                                  &\leq \|\balpha\|_2\left\|\bSigma_k-\hat\bSigma_k^{n(l)}\right\|_2\|\balpha\|_2\\
                                                                                                  &=\left\|\bSigma_k-\hat\bSigma_k^{n(l)}\right\|_2\\
                                                                                                  & \leq \left\|\bSigma_k-\hat\bSigma_k^{n(l)}\right\|_F, 
\end{align*}
which has the following consequences:
\begin{enumerate}
\item Since \(\|\bSigma_k-\hat\bSigma_k^{n(l)}\|_F\leq \lambda/2\),
  \(\balpha^{\mathsf T}\hat\bSigma_k^{n(l)}\balpha\geq
  \balpha^{\mathsf T}\bSigma_k\balpha - \lambda/2 \geq
  \lambda^{\mathrm{min}}_k - \lambda/2 \geq \lambda/2\). This implies
  \(\hat\sigma^n_k(\balpha) \geq \sqrt{\lambda/2}\).
\item Since \(\|\bSigma_k-\hat\bSigma_k^{n(l)}\|_F \to 0\) as
  \(l\to\infty\), \(\balpha^{\mathsf T}\hat\bSigma_k^{n(l)}\balpha\)
  converges to \(\balpha^{\mathsf T}\bSigma_k\balpha\)
  uniformly. \(\sqrt{\cdot}\) is uniformly continuous, so
  \(\hat\sigma_k^{n(l)}(\balpha)\) converges to \(\sigma_k(\balpha)\)
  uniformly by Lemma~\ref{lem:1}.
\end{enumerate}
Similarly, we can prove \(\hat m_k^{n(l)}(\balpha)\) converges to
\(m_k(\balpha)\) uniformly.

Let \(f=(m_0,m_1,\sigma_0,\sigma_1)\) and
\(f^l = (\hat{m}_0^{n(l)}, \hat{m}_1^{n(l)}, \hat{\sigma}_0^{n(l)},
\hat{\sigma}_1^{n(l)})\) be functions from compact
\(\mathbb{S}^{p - 1}\) to complete
\(\R^2 \times [\sqrt{\lambda/2}, \infty)^2\). We have proved that 
\(f^l \rightrightarrows f\) as \(l\to\infty\), so we can apply
Lemma~\ref{lem:2} and conclude \(\hat E^{n(l)} \rightrightarrows E\)
on \(\mathbb{S}^{p-1}\). Since
\(\sup_{[\balpha]\in \mathbb{P}^{p - 1}}|\hat E^{n(l)}([\balpha]) -
E([\balpha])| = \sup_{\balpha\in \mathbb{S}^{p -
    1}}|\hat{E}^{n(l)}(\balpha) - E(\balpha)|\to 0\), we also have
\(\hat E^{n(l)} \rightrightarrows E\) as functions on
\(\mathbb{P}^{p - 1}\). By Lemma~\ref{lem:3},
\(\hat \balpha_0^{n(l)}\to\balpha_0\) and thus
\(\hat \balpha_0^n(\omega)\to\balpha_0\). Recall this is true for any
\(\omega\in\Omega_i\) and any \(i\in\mathbb{N}\), so we have
\[
  \hat \balpha_0^n(\omega)\to\balpha_0, \;\forall
\omega\in\bigcup_{i\in\mathbb{N}}\Omega_i.
\]
Clearly, \(P((\bigcup_{i\in\mathbb{N}}\Omega_i)^c) = 0\). As a result,
\(\hat \balpha_0^n\xrightarrow{a.s.}\balpha_0\).

\subsection{Proof of Proposition~\ref{prop:2}}
\label{sec:proof-proposition-2}
For \(k = 0, 1\), since \(\tilde\bx^i_k = \bb + \bA\bx^i_k\), we have
\(\hat{\tilde\bmu}_k = (1/n_k)\sum_{i = 1}^{n_k}\tilde\bx^i_k = \bb +
(1/n_k)\bA\sum_{i = 1}^{n_k}\bx^i_k = \bb + \bA\hat{\bmu}_k\), and
\(\hat{\tilde\bSigma}_k = (1/(n_k - 1))\sum_{i = 1}^{n_k}(\tilde\bx^i_k -
\hat{\tilde\bmu}_k)(\tilde\bx^i_k - \hat{\tilde\bmu}_k)^{\mathsf{T}} = (1/(n_k -
1))\sum_{i = 1}^{n_k}\bA(\bx^i_k - \hat{\bmu}_k)(\bx^i_k -
\hat{\bmu}_k)^{\mathsf{T}} \bA^{\mathsf{T}} = \bA\hat{\bSigma}_k
\bA^{\mathsf{T}}\).

Given any directions \(\balpha\) and
\(\tilde\balpha = (\bA^{\sT})^{-1}\balpha\), for \(k = 0, 1\),
\(\hat{\tilde m}_k = \tilde\balpha^{\mathsf{T}}\hat{\tilde\bmu}_k =
\balpha^{\mathsf{T}}\bA^{-1}(\bb + \bA\hat{\bmu}_k) =
\balpha^{\mathsf{T}}\bA^{-1}\bb + \hat{m}_k\),
\(\hat{\tilde\sigma}_k =
\tilde\balpha^{\mathsf{T}}\hat{\tilde\bSigma}_k\tilde\balpha =
\balpha^{\mathsf{T}}\bA^{-1}(\bA\hat{\bSigma}_k \bA^{\mathsf{T}})
(\bA^{\sT})^{-1}\balpha = \hat{\sigma}_k\). Thus,
\(\hat{\tilde m}_0 - \hat{\tilde m}_1 = \hat{m}_0 - \hat{m}_1\), and
this implies that \(\hat{E}(\balpha) = \hat{\tilde E}(\tilde\balpha)\)
by equation~\eqref{eq:3}. In other words, \(\hat{E}\) and
\(\hat{\tilde E}\) only differ by a nonsingular linear transformation
of the domain, defined by \((\bA^{\sT})^{-1}\).

By assumption, \([\hat{\balpha}_0]\) and \([\hat{\tilde\balpha}_0]\)
are unique minimizers of \(\hat{E}\) and \(\hat{\tilde E}\)
respectively, so we have
\([\hat{\balpha}_0] = [(\bA^{\sT})^{-1}\hat{\tilde\balpha}_0]\). As a
result, there exists a constant \(c \neq 0\) such that
\(\hat{\balpha}_0 = c(\bA^{\sT})^{-1}\hat{\tilde\balpha}_0\).

\subsection{Coordinate Descent Algorithm}
\label{sec:coord-desc-algor}
Assume \(f\) is a function on \(\R^p\). Given an initial
\(\bx_0\in\mathbb{S}^{p-1}\subset\R^p\), a prefixed number of maximal
iterations \(m > 0\) and a tolerance level \(\epsilon > 0\), the
coordinate descent algorithm adapted for our method is described as
the following:
\begin{algorithm}[H]
  \label{alg:2}
  \caption{Coordinate Descent}
  \begin{algorithmic}[1]
    \Procedure{main}{$f,\bx_0,m, \epsilon$}
        \State \(i \gets 0\)
        \Repeat
          \State \(\bx_{i + 1} \gets\)
          \textsc{one\_iter\_coordinate\_descent}(\(f, \bx_i\))
          \State \(i \gets i + 1\)
        \Until{\(i= m\) \textbf{or}
          \(|f(\bx_i) - f(\bx_{i - 1})| \leq \epsilon\)}
        \State
        \Return \(\bx_i\) and \(f(\bx_i)\)
    \EndProcedure
  \end{algorithmic}
  \hrulefill
  \begin{algorithmic}[1]
    \Procedure{one\_iter\_coordinate\_descent}{$g, \by$}
      \State \(p \gets\) length of \(\by\)
      \For{\(j \gets 1, \dots, p\)}
        \State
        \(g_j(\ast) \gets g(y_1, \dots, y_{j - 1}, \ast, y_{j + 1}, \dots,
        y_p)\)
        \State \(y_j \gets\)
        \textsc{one\_dim\_coordinate\_descent}(\(g_j, y_j\))
      \EndFor
      \State \(\by \gets \by/\|\by\|_2\)
      \State \Return \(\by\)
    \EndProcedure
  \end{algorithmic}
  \hrulefill
  \begin{algorithmic}[1]
    \Procedure{one\_dim\_coordinate\_descent}{$h, t$}
      \State \(h_t \gets\) a quadratic approximation of \(h\) at
      \(t\)
      \If{\(h_t\) is concave up}
        \State \(t \gets \argmin h_t\)
      \ElsIf{\(h\) is increasing at \(t\)}
        \State \(t \gets t - 0.1\)
      \Else
        \State \(t \gets t + 0.1\)
      \EndIf
      \State \Return \(t\)
    \EndProcedure
  \end{algorithmic}
\end{algorithm}

Remark 1. Empirically, the quadratic approximation $h_t$ is not always
concave up when we update each coordinate. If it is concave down, we
update the coordinate by adding or subtracting a fixed step size of
0.1 to avoid saddle points.

Remark 2. It is possible that the sample covariance matrices
\(\hat{\bSigma}_0\), \(\hat{\bSigma}_1\) are singular. We add a small
scalar matrix (e.g. \(10^{-7}\bI_p\)) to \(\hat{\bSigma}_0\) and
\(\hat{\bSigma}_1\).

\bibliographystyle{biometrika} \bibliography{reference}
 
\end{document}